\documentclass[11pt, reqno]{amsart}

\usepackage{amsmath, amsfonts, amssymb, amsbsy, bigstrut, graphicx, enumerate,  upref, longtable, comment}

\usepackage{pstricks}

\usepackage[numbers, sort&compress]{natbib} 

\usepackage[breaklinks]{hyperref}

\newcommand{\ep}{\epsilon}

\newtheorem{thm}{Theorem}[section]
\newtheorem{lmm}[thm]{Lemma}
\newtheorem{cor}[thm]{Corollary}

\theoremstyle{definition}

\newcommand{\ee}{\mathbb{E}}

\newcommand{\ml}{\mathcal{L}}

\newcommand{\rr}{\mathbb{R}}
\newcommand{\smallavg}[1]{\langle #1 \rangle}

\newcommand{\var}{\mathrm{Var}}

\newcommand{\zz}{\mathbb{Z}}

\newcommand{\fpar}[2]{\frac{\partial #1}{\partial #2}}

\newcommand{\bg}{\mathbf{g}}

\newcommand{\mb}{\mathcal{B}}

\numberwithin{equation}{section}

\newcommand{\eq}[1]{\begin{align*} #1 \end{align*}}

\begin{document}
\title[Decay of correlations in RFIM]{On the decay of correlations in the  random field Ising model}
\author{Sourav Chatterjee}
\address{\newline Department of Statistics \newline Stanford University\newline Sequoia Hall, 390 Serra Mall \newline Stanford, CA 94305\newline \newline \textup{\tt souravc@stanford.edu}}
\thanks{Research partially supported by NSF grant DMS-1608249}

\begin{abstract}
In a celebrated 1990 paper, Aizenman and Wehr proved that the two-dimensional random field Ising model has a unique infinite volume Gibbs state at any temperature. The proof is ergodic-theoretic in nature and does not provide any quantitative information. This article proves the first quantitative version of the Aizenman--Wehr theorem. The proof introduces a new method for proving decay of correlations that may be interesting in its own right. A fairly detailed sketch of the main ideas behind the proof is also included.
\end{abstract}

\maketitle

\section{Introduction}
Let $\Lambda$ be a finite subset of $\zz^d$. Let $\partial \Lambda$ be the set of all $x\in \zz^d\setminus \Lambda$ that are adjacent to some $y\in \Lambda$. We will refer to $\partial \Lambda$ as the outer boundary (or simply the boundary) of $\Lambda$. Let $\Sigma=\{-1,1\}^\Lambda$ and $\Gamma = \{-1,1\}^{\partial\Lambda}$. An element of $\Sigma$ will be called a configuration and an element of $\Gamma$ will be called a boundary condition. Let $\Phi = \rr^\Lambda$. Elements of $\Phi$ will be called external fields. For $\sigma \in \Sigma$, $\gamma \in \Gamma$ and $\phi\in \Phi$, define the energy of $\sigma$ as
\[
H_{\gamma, \phi}(\sigma) := -\frac{1}{2}\sum_{\substack{x,y\in \Lambda,\\x\sim y}} \sigma_x\sigma_y - \sum_{\substack{x\in \Lambda, \, y\in \partial \Lambda,\\ x\sim y}}\sigma_x\gamma_y- \sum_{x\in \Lambda} \phi_x \sigma_x,
\]
where $x\sim y$ means that $x$ and $y$ are neighbors. Take any $\beta\in [0,\infty]$. The Ising model on $\Lambda$ with boundary condition $\gamma$, inverse temperature $\beta$, and external field $\phi$, is the probability measure on $\Sigma$ with probability mass function proportional to $e^{-\beta H_{\gamma, \phi}(\sigma)}$. When $\beta=\infty$, this is simply the uniform probability measure on the configurations that minimize the energy (the ground states). 

Let us now suppose that $(\phi_x)_{x\in \Lambda}$ are i.i.d.~random variables instead of fixed constants. Then the probability measure defined above becomes a random probability measure. This is known as the random field Ising model (sometimes abbreviated as RFIM). We will refer to the law of $\phi_x$ as the random field distribution.

The random field Ising model was introduced by \citet{imryma75} as a simple example of a disordered system. Imry and Ma predicted that the model does not have an ordered phase in dimensions one and two, but does exhibit a phase transition in dimensions three and higher. Under some conditions on the random field distribution, \citet{bk87, bk88} settled the Imry--Ma conjecture in $d\ge 3$, and \citet{aw89, aw90} settled it in $d\le 2$. For a readable account of these proofs and an up-to-date survey of the literature, see \citet[Chapter 7]{bovier06}.

An important consequence of the Aizenman--Wehr theorem is that the 2D RFIM exhibits decay of correlations at any temperature. One way to state this precisely is the following. Let all notation be as in the beginning of this section, and take any $x\in \Lambda$. Choose any random field distribution, and consider the RFIM on $\Lambda$ at some inverse temperature $\beta\in [0,\infty]$ and some boundary condition $\gamma\in \Gamma$. Let $\smallavg{\sigma_x}_\gamma$ denote the quenched expected value of $\sigma_x$ in this model. Decay of correlations means that
\[
\sup_{\gamma, \gamma'\in \Gamma} |\smallavg{\sigma_x}_\gamma - \smallavg{\sigma_x}_{\gamma'}| \to 0
\]
in probability as $\Lambda\uparrow \zz^2$, with $x$ and $\beta$ remaining fixed. In other words, the effect of the boundary condition on the law of the spin at some interior point becomes negligible as the distance of the point from the boundary becomes large. Under mild conditions on the random field distribution, this result follows from the Aizenman--Wehr theorem, and is in fact equivalent to it. The proof of the Aizenman--Wehr theorem, however, uses ergodic theory in a crucial way and provides no quantitative information. The question of establishing a rate for the decay of correlations in the 2D RFIM has remained open, except at sufficiently small $\beta$ where standard techniques can be used to prove exponential decay. The following theorem gives the first rate of decay at arbitrary $\beta$. 
\begin{thm}\label{mainthm}
Consider the random field Ising model on a set $\Lambda\subseteq \zz^2$ at inverse temperature $\beta \in [0,\infty]$, as defined in the beginning of this section. Let the random field distribution be Gaussian with mean zero and variance~$v$. Take any $x\in \Lambda$ such that $n\ge 3$, where $n$ is the $\ell^\infty$ distance of $x$ from $\partial \Lambda$. Then
\[
\ee\biggl(\sup_{\gamma, \gamma'\in \Gamma} |\smallavg{\sigma_x}_\gamma - \smallavg{\sigma_x}_{\gamma'}| \biggr) \le\frac{ C(1+v^{-1/2})}{\sqrt{\log \log n}},
\]
where $C$ is a universal constant. In particular, the bound has no dependence on $\beta$ and holds even if $\beta=\infty$.
\end{thm}
The above theorem gives quantitative information on how the quenched law of the spin at a single site depends on the boundary condition. There remains, of course, the possibility that the rate can be improved. There is a folklore conjecture that the true rate of decay is exponentially fast in $n$ at any $\beta$. There is also a competing belief that the rate may be polynomial in $n$  at large $\beta$. Proving either of these conjectures would be a substantial improvement of Theorem \ref{mainthm}. Using the approach of this paper, however,  I do not see any way of getting a better rate than the one given in Theorem~\ref{mainthm}. Any improvement will need a new idea. 

Another way to improve Theorem \ref{mainthm} is by extending it to non-Gaussian random field distributions. Again, the proof in this paper uses the Gaussianity quite heavily, to the extent that I do not see any obvious way to adapt it to a non-Gaussian setting.

\section{Sketch of the proof}
Since the proof of Theorem \ref{mainthm} does not follow any of the standard techniques for proving correlation decay, and is also quite different than the approach of \citet{aw90}, it may be worthwhile to explain the main ideas here, before embarking on the details. The ideas may be applicable to other disordered systems. Unfortunately, I have found it hard to encapsulate the scheme in a few paragraphs, so the sketch itself is a few pages long.

Throughout, $C$ will denote any universal constant. For simplicity, we will assume that $\Lambda$ is an $n\times n$ square, $x$ is the center of the square, and $v=1$. By the well-known FKG  property of the RFIM, $\smallavg{\sigma_x}_\gamma$ is a monotone increasing function of the boundary condition $\gamma$. Therefore, it suffices to show that
\[
 \ee(\smallavg{\sigma_x}_+ - \smallavg{\sigma_x}_-)\le \frac{C}{\sqrt{\log\log n}},
 \]
  where $+$ and $-$ denote the boundary conditions in which all boundary spins are $+1$ and $-1$, respectively. It turns out that by a simple translation invariance argument, it suffices to show that
\begin{equation}\label{mpmm}
\ee(M_+-M_-)\le \frac{Cn^2}{\sqrt{\log \log n}},
\end{equation}
where 
\[
M_+:= \sum_{x\in \Lambda}\smallavg{\sigma_x}_+ \ \text{ and } \ M_-:=\sum_{x\in \Lambda} \smallavg{\sigma_x}_-.
\]
Let $m\ll n$ be a number, to be chosen later. Partition $\Lambda$ into a collection $\mathcal{B}$ of $m\times m$ sub-squares. For each $B\in \mathcal{B}$, let
\[
M_+(B) := \sum_{x\in B}\smallavg{\sigma_x}_+, \ \text{ and } \ M_-(B) := \sum_{x\in B}\smallavg{\sigma_x}_-.
\]
We will show that for most $B\in \mathcal{B}$, 
\begin{equation}\label{meq}
\ee(M_+(B))-\ee(M_-(B))\le \frac{Cm^2}{\sqrt{\log \log n}}.
\end{equation}
Summing over $B$ (assuming that the set of exceptional $B$ is small enough), this  proves \eqref{mpmm}.

Let $F_\gamma$ be the free energy ($=$ logarithm of the partition function) of the model under boundary condition $\gamma$. Fix an $m\times m$ square $B$. Modify the model by replacing $\phi_x$ with $\phi_x+h$ for all $x\in B$, but keeping all other $\phi_x$ the same. Let $F_\gamma(h)$ be the new free energy. This quantity is useful because
\begin{equation}\label{mrep}
M_+(B) = \frac{1}{\beta}F'_+(0) \ \text{ and } \ M_-(B) = \frac{1}{\beta}F'_-(0).
\end{equation}
So, to prove \eqref{meq}, we need to show that
\[
\ee(F_+'(0)) - \ee(F_-'(0)) \le \frac{C\beta m^2}{\sqrt{\log \log n}}. 
\]
We will show this by approximating $F_\pm'(0)$ with $(F_\pm(h)-F_\pm(0))/h$ for some suitable small $h$.

Take any boundary condition $\gamma$ and any $h$, and consider the modified model defined above. Slightly tweak this model by decoupling the links between $B$ and $\Lambda\setminus B$. Let $G_\gamma(h)$ be the free energy of the new model. Due to the decoupling, $G_\gamma(h)$ decomposes as a sum of contributions from inside and outside $B$. The contribution from outside $B$ does not depend on $h$, and the contribution from inside $B$ does not depend on $\gamma$.  Thus, there is some $\alpha(h)$ depending only on $h$ and not on $\gamma$, such that
\[
G_\gamma(h)-G_\gamma(0) = \alpha(h).
\]
We will show that for any $h$ and $\gamma$,
\[
|F_\gamma(h)-G_\gamma(h)|\le 4\beta m.
\]
(Briefly, this holds because  $d=2$ and $|\partial B|\le 4m$.) Combining, we get  that for any $\gamma$,
\[
|(F_\gamma(h)-F_\gamma(0))-\alpha(h)|\le Cm.
\] 
Consequently,
\begin{equation}\label{sketchbd1}
\biggl|\frac{F_+(h)-F_+(0)}{h} - \frac{F_-(h)-F_-(0)}{h}\biggr|\le \frac{C\beta m}{h}.
\end{equation}
It only remains to get a bound for
\begin{equation}\label{sketchbd2}
\biggl|\ee(F_\pm'(0)) - \ee\biggl(\frac{F_\pm(h)-F_\pm(0)}{h}\biggr)\biggr|
\end{equation}
in terms of $m$ and $h$, and then show that $m$ and $h$ can be chosen so that the quantities in both of the above displays are bounded by $C\beta m^2/\sqrt{\log \log n}$.  By \eqref{mrep}, this will complete the proof of \eqref{meq}. We will now sketch this step for the plus boundary condition, the argument being similar for minus boundary. 

We will start by rigorously justifying the Taylor expansion
\[
\ee\biggl( \frac{F_+(h)-F_+(0)}{h}\biggr) - \ee(F_+'(0)) = \sum_{k=2}^\infty \frac{h^{k-1}}{k!} \ee(F^{(k)}_+(0)).
\]
Suppose that this has been justified. It is not hard to see that
\[
F_+^{(k)}(0) = \sum_{x_1,\ldots,x_k\in B} \frac{\partial^k F_+}{\partial \phi_{x_1}\cdots \partial \phi_{x_k}},
\]
where the $F_+$ on the right denotes the free energy of the original model under plus boundary. 
Thus,
\begin{align}\label{remterm}
\sum_{k=2}^\infty \frac{h^{k-1}}{k!} \ee(F^{(k)}_+(0)) &= \sum_{k=2}^\infty \sum_{x_1,\ldots,x_k\in B} \frac{h^{k-1}}{k!}\ee\biggl(\frac{\partial^k F_+}{\partial \phi_{x_1}\cdots \partial \phi_{x_k}}\biggr). 
\end{align}
We will now sketch how to bound this remainder term. Since $F_+$ is a function of standard Gaussian random variables, we can write its $L^2$ norm using the Fourier expansion of $F_+$ in the multivariate Hermite polynomial basis of Gaussian $L^2$ space.  It turns out that the quantities  
\[
\frac{1}{\sqrt{k!}}\ee\biggl(\frac{\partial^k F_+}{\partial \phi_{x_1}\cdots \partial \phi_{x_k}}\biggr),
\]
as $x_1,\ldots,x_k$ range over $\Lambda$,  are its Fourier coefficients. In particular,
\[
\var(F_+) = \sum_{k=1}^\infty \sum_{x_1,\ldots, x_k\in \Lambda}\frac{1}{k!}\biggl(\ee\biggl(\frac{\partial^k F_+}{\partial \phi_{x_1}\cdots \partial \phi_{x_k}}\biggr)\biggr)^2.
\]
Using the Gaussian Poincar\'e inequality, we will show that 
\[
\var(F_+)\le C\beta^2n^2.
\]
Combining the above two displays gives 
\[
\sum_{k=1}^\infty \sum_{x_1,\ldots, x_k\in \Lambda}\frac{1}{k!}\biggl(\ee\biggl(\frac{\partial^k F_+}{\partial \phi_{x_1}\cdots \partial \phi_{x_k}}\biggr)\biggr)^2\le C\beta^2n^2.
\]
From this and Markov's inequality, it follows that if $K$ is large, then for most $B\in \mb$,
\[
\sum_{k=1}^\infty \sum_{x_1,\ldots, x_k\in B}\frac{1}{k!}\biggl(\ee\biggl(\frac{\partial^k F_+}{\partial \phi_{x_1}\cdots \partial \phi_{x_k}}\biggr)\biggr)^2\le K^2\beta^2m^2.
\]
Suppose that our chosen $B$ is one such square. Then it is natural to think about bounding the right side of \eqref{remterm} using the above bound and the Cauchy--Schwarz inequality. A straightforward application of Cauchy--Schwarz gives
\begin{align}
&\biggl|\sum_{k=2}^\infty \sum_{x_1,\ldots,x_k\in B} \frac{h^{k-1}}{k!}\ee\biggl(\frac{\partial^k F_+}{\partial \phi_{x_1}\cdots \partial \phi_{x_k}}\biggr)\biggr|\nonumber\\
&\le \biggl(\sum_{k=2}^\infty \frac{h^{2k-2}m^{2k}}{k!}\biggr)^{1/2}\biggl(\sum_{k=2}^\infty \sum_{x_1,\ldots,x_k\in B} \frac{1}{k!}\biggl(\ee\biggl(\frac{\partial^k F_+}{\partial \phi_{x_1}\cdots \partial \phi_{x_k}}\biggr)\biggr)^2\biggr)^{1/2}\nonumber\\
&\le K\beta m^2e^{hm}. \label{badbd}
\end{align}
Using this as an upper bound in \eqref{sketchbd2} and combining with \eqref{sketchbd1} and \eqref{mrep}, we get
\[
\ee(M_+(B))-\ee(M_-(B)) \le \frac{Cm}{h} + Km^2e^{hm}. 
\]
We would like to choose a small $h$ and a  large $K$ so that this bound is better than the trivial bound $Cm^2$. Unfortunately, there is no way to choose such $h$ and $K$. The best we can do with the above bound is, in  fact, $Cm^2$.

Note that until now, we have used no special property of $m$. Indeed, $m$ could as well have been equal to $n$. To complete the proof, we will show that $m$ can be chosen in a clever way that allows a suitable improvement of our Cauchy--Schwarz step. The rest of this section sketches this step. 
 
For any $x_1,\ldots, x_k\in \rr^2$, let
\[
d(x_1,\ldots,x_k):= \max_{1\le i,j\le k} |x_i-x_j|_\infty,
\]
where $|x|_\infty$ denotes the $\ell^\infty$ norm of $x$. For each $j$, let 
\[
s_j := \sum_{k=1}^\infty \sum_{\substack{x_1,\ldots, x_k\in \Lambda,\\ (\log n)^{j-1}\le d(x_1,\ldots,x_k)< (\log n)^j}}\frac{1}{k!}\biggl(\ee\biggl(\frac{\partial^k F_+}{\partial \phi_{x_1}\cdots \partial \phi_{x_k}}\biggr)\biggr)^2.
\]
Then 
\[
\sum_{j=1}^\infty s_j\le \var(F_+)\le C\beta^2n^2.
\] 
Thus, for any $j$, there exists $i\le j$ such that 
\[
s_i\le \frac{C\beta^2n^2}{j}.
\]
In particular, there exists 
\[
i\le \frac{\log n}{2\log\log n}
\]
such that
\[
s_{i} \le \frac{C\beta^2n^2\log \log n}{\log n}.
\]
As before, this can be used to show that if $K$ is large, then for most $B\in \mb$, 
\begin{align}
&\sum_{k=2}^\infty\sum_{\substack{x_1,\ldots, x_k\in B,\\ (\log n)^{i-1}< d(x_1,\ldots,x_k)\le (\log n)^{i}}} \frac{1}{k!}\biggl(\ee\biggl(\frac{\partial^k F_+}{\partial \phi_{x_1}\cdots \partial \phi_{x_k}}\biggr)\biggr)^2\nonumber\\
&\qquad \le \frac{K\beta^2m^2\log \log n}{\log n}.\label{sismall}
\end{align}
Let $m$ be the integer part of $(\log n)^{i}$, so that $m\le \sqrt{n}$. Then 
\begin{align*}
&\sum_{k=2}^{\infty} \frac{h^{k-1}}{k!}\ee(F_+^{(k)}(0)) =  \sum_{k=2}^\infty  \sum_{x_1,\ldots,x_k\in B} \frac{h^{k-1}}{k!}\ee\biggl(\frac{\partial^k F_+}{\partial \phi_{x_1}\cdots \partial \phi_{x_k}}\biggr)\\
&= \sum_{k=2}^\infty \sum_{\substack{x_1,\ldots, x_k\in B,\\ d(x_1,\ldots,x_k)\le (\log n)^{i-1}}} \frac{h^{k-1}}{k!}\ee\biggl(\frac{\partial^k F_+}{\partial \phi_{x_1}\cdots \partial \phi_{x_k}}\biggr) \\
&\qquad + \sum_{k=2}^\infty\sum_{\substack{x_1,\ldots, x_k\in B,\\ (\log n)^{i-1}< d(x_1,\ldots,x_k)\le (\log n)^{i}}} \frac{h^{k-1}}{k!}\ee\biggl(\frac{\partial^k F_+}{\partial \phi_{x_1}\cdots \partial \phi_{x_k}}\biggr).
\end{align*} 
Separately apply Cauchy--Schwarz to the two parts, and then apply \eqref{sismall} to the second part.   This gives 
\begin{align*}
\biggl|\sum_{k=2}^{\infty} \frac{h^{k-1}}{k!}\ee(F_+^{(k)}(0))\biggr|&\le C\beta m^2 \biggl(\frac{hm}{\log n} + Ke^{h^2m^2}\sqrt{\frac{\log\log n}{\log n}}\biggr).
\end{align*}
This is an improvement of \eqref{badbd}, since it allows choices of $h$ and $K$ such that the right side is $o(m^2)$. 
The proof of \eqref{meq} is now easily completed by choosing $h = \sqrt{\log \log n}/2m$ and $K$ to be a small power of $\log n$. 

\section{Proof details}\label{proofs}
This section contains the detailed proof of Theorem \ref{mainthm}. A key ingredient in the proof of Theorem \ref{mainthm} is the following formula for the variance of a function of independent standard Gaussian random variables. 
\begin{thm}[\cite{cha09}]\label{varformula}
Let $\bg = (g_1,\ldots, g_n)$ be a vector of i.i.d.\ standard Gaussian random variables, and let $f$ be a $C^\infty$ function of $\bg$ with bounded derivatives of all orders. Then 
\begin{equation}\label{varid}
\var(f) = \sum_{k=1}^\infty \frac{1}{k!}\sum_{1\le i_1,\ldots,i_k\le n} \biggl(\ee\biggl(\frac{\partial^k f}{\partial g_{i_1}\cdots \partial g_{i_k}}\biggr)\biggr)^2. 
\end{equation}
The convergence of the infinite series is part of the conclusion. 
\end{thm}
Although the above version of this identity first appeared in \cite{cha09}, slightly different but equivalent versions were already present in the earlier papers~\cite{houdreperezabreu95, houdreetal98}. The identity has been used recently in \cite{chatterjee09, chabook, chatterjee15}. The proof is quite simple, and goes as follows. Let $\gamma^n$ denote the standard Gaussian measure on $\rr^n$. It is a well-known fact that the $n$-variable Hermite polynomials form an orthonormal basis of $L^2(\gamma^n)$. Using integration by parts, it is not difficult to prove that the Fourier coefficients of $f$ with respect to this orthonormal basis can be expressed as the expectations of mixed partial derivatives of $f$ occurring on the right side of~\eqref{varid}. The identity \eqref{varid} is simply the Parseval identity for this Fourier expansion.

A second ingredient in the proof of Theorem \ref{mainthm} is the Gaussian Poincar\'e inequality, stated below.
\begin{thm}[Gaussian Poincar\'e inequality]\label{pointhm}
Let $f$ and $\bg$ be as in Theorem~\ref{varformula}. Then
\begin{equation}\label{poinid}
\var(f)\le  \ee\biggl(\sum_{i=1}^n\biggl(\fpar{f}{g_i}\biggr)^2\biggr).
\end{equation}
\end{thm}
A simple proof of the Gaussian Poincar\'e inequality can be given using Theorem \ref{varformula}, by applying \eqref{varid} to each $\partial f/\partial g_i$ and then adding up the results to get an expansion for the right side of \eqref{poinid}. Comparing this expansion with the expansion for $\var(f)$ easily shows that one dominates the other. For more on the Gaussian Poincar\'e inequality and the related literature, see~\cite[Chapter 2]{chabook}.

In the remainder of this section, the term `plus boundary condition' will mean, as usual, the boundary condition $\gamma$ where each $\gamma_x = 1$. The quenched expectation of the spin  at site $x$ under plus boundary condition will be denoted by $\smallavg{\sigma_x}_+$. If the domain $\Lambda$ needs to be emphasized, we will write $\smallavg{\sigma_x}_{\Lambda, +}$. Minus boundary condition and related notations are defined similarly. An important consequence of the FKG property is that for any boundary condition~$\gamma$,
\begin{equation}\label{sigma1}
\smallavg{\sigma_x}_+\ge\smallavg{\sigma_x}_\gamma\ge \smallavg{\sigma_x}_-.
\end{equation}
From \eqref{sigma1} and the Markovian nature of the RFIM, it follows that for any $x\in\Lambda'\subseteq \Lambda$,
\begin{equation}\label{sigma2}
\smallavg{\sigma_x}_{\Lambda',+}\ge \smallavg{\sigma_x}_{\Lambda,+} \ \text{ and }\ \smallavg{\sigma_x}_{\Lambda',-}\le \smallavg{\sigma_x}_{\Lambda,-}.
\end{equation}
Throughout, we will assume that the random field distribution is Gaussian with mean zero and variance $v$. Instead of $\phi_x$,  the external field at a vertex $x$ will be denoted by $\sqrt{v} \phi_x$, where $(\phi_x)_{x\in \zz^2}$ are i.i.d.~standard Gaussian random variables. Lastly, $C$ will denote any universal constant, whose value may change from line to line.

The main step in the proof of Theorem~\ref{mainthm} is the following lemma.
\begin{lmm}\label{sqlmm}
Let $\Lambda$ be an $n\times n$ square, for some $n\ge 3$. Consider the RFIM on $\Lambda$ at inverse temperature $\beta\in (0,\infty)$. Then there exists $x\in \Lambda$ such that
\[
\ee(\smallavg{\sigma_x}_+-\smallavg{\sigma_x}_-)\le \frac{C(1+v^{-1/2})}{\sqrt{\log \log n}}. 
\]
\end{lmm}
The proof of Lemma \ref{sqlmm} is somewhat long and complicated, and is therefore divided into several steps. Throughout, fix $\beta \in (0,\infty)$ and an $n\times n$ square $\Lambda$. Let $F_+$ be the free energy ($=$ the logarithm of the partition function) of the RFIM on $\Lambda$ with plus boundary condition, at inverse temperature $\beta$. For any $k$ and any $x_1,\ldots,x_k\in \Lambda$, let
\[
\rho_+(x_1,\ldots, x_k) := \ee\biggl(\frac{\partial^kF_+}{\partial \phi_{x_1}\cdots \partial \phi_{x_k}}\biggr). 
\]
Let $\rho_-(x_1,\ldots, x_k)$ be defined analogously, for the RFIM on $\Lambda$ with minus boundary condition. The following lemma is the first step in the proof of Lemma \ref{sqlmm}. 
\begin{lmm}\label{poinstep}
Let $\rho_+$ be defined as above. Then
\begin{equation*}
 \sum_{k=1}^\infty\frac{1}{k!}\sum_{x_1,\ldots, x_k\in \Lambda}\rho_+(x_1,\ldots,x_k)^2\le \beta^2vn^2,
\end{equation*}
and the same inequality holds for  $\rho_-$ as well. 
\end{lmm}
\begin{proof}
By Theorem \ref{varformula}, 
\[
\var(F_+) = \sum_{k=1}^\infty\frac{1}{k!}\sum_{x_1,\ldots, x_k\in \Lambda}\rho_+(x_1,\ldots,x_k)^2. 
\]
On the other hand,
\[
\fpar{F_+}{\phi_x} = \beta \sqrt{v} \smallavg{\sigma_x}_+.
\]
By Theorem \ref{pointhm}, this shows that
\[
\var(F_+)\le \beta^2 v n^2. 
\]
Combining the above observations, we get the desired inequality. 
Retracing the above steps, it is clear that the inequality holds for $\rho_-$ as well.
\end{proof}

Let $B$ be a sub-square of $\Lambda$. Take any $h\in \rr$. Consider the RFIM on $\Lambda$ with plus boundary condition, and slightly tweak this model to obtain a new model by replacing $\phi_x$ with $\phi_x+h$ for each $x\in B$, keeping all other $\phi_x$ the same. Let $F_+(h)$ be the free energy of this new model, so that $F_+(0)$ is the free energy of the original model. As a function of $h$, it is easy to check that $F_+(h)$ is infinitely differentiable. Let  $F_+^{(k)}$ denote the $k^{\textup{th}}$ derivative of $F_+$. The following Taylor series expansion for the expected value of $F_+(0)$ is the second step in the proof of Lemma~\ref{sqlmm}. The convergence of the series in this lemma is a nontrivial claim, because a direct computation of the $k^{\mathrm{th}}$ derivative yields an expression with a super-exponentially growing number of terms.
\begin{lmm}\label{taylorstep}
Let $F_+(h)$ be defined as above. Then for any $h\ge 0$,
\[
\ee(F_+(h)) = \ee(F_+(0))+\sum_{k=1}^{\infty} \frac{h^k}{k!}\ee(F_+^{(k)}(0)).
\]
\end{lmm}
\begin{proof}
In the following, $m$ will denote the width of $B$. For $x_1,\ldots,x_k\in \Lambda$, let
\[
\rho_{+,h}(x_1,\ldots, x_k) := \ee\biggl(\frac{\partial^kF_+(h)}{\partial \phi_{x_1}\cdots \partial \phi_{x_k}}\biggr). 
\]
Proceeding exactly as in the proof of Lemma \ref{poinstep}, we get that for any $h$, 
\begin{equation}\label{ineq2}
 \sum_{k=1}^\infty\frac{1}{k!}\sum_{x_1,\ldots, x_k\in \Lambda}\rho_{+,h}(x_1,\ldots,x_k)^2\le \beta^2vn^2. 
\end{equation}
But note that 
\begin{equation}\label{fplus}
F_+^{(k)}(h) = \sum_{x_1,\ldots, x_k\in B} \frac{\partial^kF_+(h)}{\partial \phi_{x_1}\cdots \partial \phi_{x_k}}.
\end{equation}
Therefore by the Cauchy--Schwarz inequality and \eqref{ineq2}, for any nonnegative $h$ and $u$, 
\begin{align*}
&\sum_{k=1}^\infty \frac{h^{k-1}|\ee(F_+^{(k)}(u))|}{(k-1)!} \le \sum_{k=1}^\infty \sum_{x_1,\ldots, x_k\in B} \frac{kh^{k-1}}{k!} |\rho_{+,u}(x_1,\ldots,x_k)|\\
&\le \biggl(\sum_{k=1}^\infty \frac{k^2h^{2k-2}m^{2k}}{k!}\biggr)^{1/2}\biggl(\sum_{k=1}^\infty\sum_{x_1,\ldots,x_k\in B} \frac{1}{k!}\rho_{+,u}(x_1,\ldots,x_k)^2\biggr)^{1/2}\\
&\le \beta\sqrt{v}n C(m,h),
\end{align*}
where $C(m,h)$ is a finite real number that depends only on $m$ and $h$. Note that the bound has no dependence on $u$. Thus, for any $h\ge 0$,
\begin{align*}
\sum_{k=1}^\infty \int_0^h\frac{(h-u)^{k-1}}{(k-1)!} |\ee(F_+^{(k)}(u))| du &\le \int_0^h \sum_{k=1}^\infty\frac{h^{k-1}}{(k-1)!} |\ee(F_+^{(k)}(u))| du\\
&\le \beta \sqrt{v} n C(m,h)h <\infty. 
\end{align*}
This shows, in particular, that
\begin{equation}\label{flim}
\lim_{k\to \infty}\int_0^h\frac{(h-u)^{k-1}}{(k-1)!} \ee(F_+^{(k)}(u)) du  = 0.
\end{equation}
But Taylor expansion gives
\[
F_+(h) = F_+(0)+\sum_{j=1}^{k-1} \frac{h^j}{j!}F_+^{(j)}(0) + \int_0^h\frac{(h-u)^{k-1}}{(k-1)!} F_+^{(k)}(u) du. 
\]
By \eqref{flim}, the expectation of the remainder term goes to zero as $k\to \infty$. This gives the desired result.
\end{proof}

The sub-square $B$ in Lemma \ref{taylorstep} is arbitrary. We will now choose a specific sub-square $B$. Let $\ep:= 1/\log n$ and let $m_i := \ep^{-i}$ for $i\ge 1$. Let $m_0=0$. For any $k$ and any $x_1,\ldots,x_k\in \Lambda$, let
\[
d(x_1,\ldots,x_k):= \max_{1\le p<q\le k} |x_p-x_q|_{\infty},
\]
where $|x|_\infty$ denotes the $\ell^\infty$ norm of a vector $x\in \rr^2$. 
For each $i\ge 1$, let 
\[
s_i := \sum_{k=1}^\infty\frac{1}{k!}\sum_{\substack{x_1,\ldots,x_k\in \Lambda,\\  m_{i-1}\le d(x_1,\ldots,x_k)< m_i}}( \rho_+(x_1, \ldots,x_k)^2 +  \rho_-(x_1, \ldots,x_k)^2).
\]
Then by Lemma \ref{poinstep}, 
\begin{equation}\label{sisum}
\sum_{i=1}^\infty s_i \le 2\beta^{2}vn^2. 
\end{equation}
Let $L$ be the smallest integer for which $m_{L}\ge \sqrt{n}$. By the above inequality, there exists $i$ such that $1\le i\le L$ and 
\begin{equation}\label{sibd}
s_i\le \frac{2\beta^{2}vn^2}{L}\le 4\beta^{2}vn^2 \frac{\log\log n}{\log n}. 
\end{equation}
Fix such an $i$. Let $m$ be the largest integer that is strictly less than $m_i$. Since $m_i\ge m_1 =\log n>1$, it follows that $m\ge 1$. Let $\Lambda_0$ be a sub-square of $\Lambda$ with side-length $[n/m]m$. Note that
\begin{equation}\label{l0size}
|\Lambda \setminus\Lambda_0|\le 2nm\le 2nm_L\le 2n^{3/2}\log n. 
\end{equation}
Partition $\Lambda_0$ into a collection $\mb$ of $m\times m$ sub-squares in the natural way. For each $B\in \mb$, let 
\[
s_0(B) :=\sum_{k=1}^\infty\frac{1}{k!} \sum_{\substack{x_1,\ldots, x_k\in B,\\ d(x_1,\ldots,x_k)<m_{i-1} } } (\rho_+(x_1,\ldots,x_k)^2+\rho_-(x_1,\ldots, x_k)^2),
\]
and let 
\[
s_1(B) := \sum_{k=1}^\infty \frac{1}{k!}\sum_{\substack{x_1,\ldots,x_k\in B,\\ m_{i-1}\le d(x_1,\ldots,x_k)< m_{i}} } (\rho_+(x_1,\ldots,x_k)^2+\rho_-(x_1,\ldots, x_k)^2).
\]
Notice that
\eq{
|\mb| &= \biggl[\frac{n}{m}\biggr]^2 \ge \frac{n^2}{4m^2}. 
}
Thus, by \eqref{sibd}, 
\begin{equation}\label{s1bd}
\bar{s}_1:= \frac{1}{|\mb|}\sum_{B\in \mb} s_1(B)\le \frac{s_i}{|\mb|}\le 16\beta^{2}v m^2 \frac{\log\log n}{\log n}, 
\end{equation}
and  by \eqref{sisum},
\begin{equation}\label{s0bd}
\bar{s}_0:= \frac{1}{|\mb|}\sum_{B\in \mb} s_0(B) \le \frac{1}{|\mb|}\sum_{j=1}^\infty s_j\le 8\beta^{2}vm^2.
\end{equation}
Let us now define 
\[
K:=(\log n)^{1/12}.
\] 
The value of $K$ will remain fixed throughout the rest of the proof. 
Let $\mb_0$ be the set of all $B\in \mb$ such that $s_1(B)\le K^2\bar{s}_1$ and $s_0(B)\le K^2\bar{s}_0$. Then by Markov's inequality,
\begin{equation}\label{mb1}
|\mb\setminus \mb_0|\le \frac{2|\mb|}{K^2}.
\end{equation}
The third step in the proof of Lemma \ref{sqlmm} is the following estimate for $B\in \mb_0$.
\begin{lmm}\label{remainstep}
Let $\mb_0$, $m$ and $K$ be as above. Fixing a choice of $B\in \mb_0$, define $F_+(h)$ as in the paragraph preceding the statement of Lemma \ref{taylorstep}. Let  
\[
h := \frac{\sqrt{\log \log n}}{2m}.
\]
Then 
\[
\biggl|\ee(F_+'(0)) - \frac{\ee(F_+(h))-\ee(F_+(0))}{h}\biggr| \le  CK\beta \sqrt{v} m^2 \frac{\sqrt{\log\log n}}{(\log n)^{1/4}},
\]
and the same bound holds for $F_-$.
\end{lmm}
\begin{proof}
By Lemma \ref{taylorstep},
\[
\biggl|\ee(F_+'(0)) - \frac{\ee(F_+(h))-\ee(F_+(0))}{h}\biggr|\le \sum_{k=2}^{\infty} \frac{h^{k-1}}{k!}|\ee(F_+^{(k)}(0))|.
\]
By \eqref{fplus} and the fact that $m<m_i$,
\begin{align}
&\sum_{k=2}^{\infty} \frac{h^{k-1}}{k!}|\ee(F_+^{(k)}(0))| \le \sum_{k=2}^\infty  \sum_{x_1,\ldots,x_k\in B} \frac{h^{k-1}}{k!}|\rho_+(x_1,\ldots, x_k)|\nonumber \\
&= \sum_{k=2}^\infty \sum_{\substack{x_1,\ldots,x_k\in B,\\ d(x_1,\ldots,x_k)<m_{i-1}}} \frac{h^{k-1}}{k!}|\rho_+(x_1,\ldots,x_k)| \nonumber \\
&\qquad + \sum_{k=2}^\infty\sum_{\substack{x_1,\ldots,x_k\in B,\\ m_{i-1}\le d(x_1,\ldots,x_k)<m_{i}}} \frac{h^{k-1}}{k!}|\rho_+(x_1,\ldots,x_k)|.\label{mainbreak}
\end{align}
The number of ways of choosing $x_1,\ldots, x_k\in B$ is $m^{2k}$. Therefore by the Cauchy--Schwarz inequality, the fact that $B\in \mb_0$, and the bound \eqref{s1bd}, we get
\eq{
&\sum_{k=2}^\infty\sum_{\substack{x_1,\ldots,x_k\in B,\\ m_{i-1}\le d(x_1,\ldots,x_k)<m_{i}}} \frac{h^{k-1}}{k!}|\rho_+(x_1,\ldots,x_k)| \\
&\le  \biggl(\sum_{k=2}^\infty \frac{h^{2k-2}m^{2k}}{k!}\biggr)^{1/2}\sqrt{s_1(B)}\le  4K\beta\sqrt{v} m^2e^{h^2m^2} \sqrt{\frac{\log\log n}{\log n}}\\
&= 4K\beta \sqrt{v} m^2 \frac{\sqrt{\log\log n}}{(\log n)^{1/4}}. 
}
If $i=1$, then there is no $x_1,\ldots,x_k\in B$ such that $d(x_1,\ldots,x_k)<m_{i-1}$. Therefore, in this case, the first sum on the right side in \eqref{mainbreak} is zero. Suppose that $i>1$. Then the number of ways of choosing $x_1,\ldots, x_k\in B$ such that $d(x_1,\ldots, x_k)< m_{i-1}$ is bounded above by $m^2 (2m_{i-1}-1)^{2(k-1)}$, since $x_1$ can be chosen in $m^2$ ways, and given $x_1$, the constraint $d(x_1,\ldots,x_k)<m_{i-1}$ implies that $x_2,\ldots,x_k$ have to be within a square of side-length $2m_{i-1}-1$ centered at $x_1$. Since $\ep = 1/\log n< 1/2$, 
\[
2m_{i-1}-1 \le 2\ep m_i -1 \le 2\ep(m+1)-1\le 2\ep m. 
\]
Thus, by the Cauchy--Schwarz inequality, the fact that $B\in \mb_0$, and the bound \eqref{s0bd}, we get
\eq{
&\sum_{k=2}^\infty \sum_{\substack{x_1,\ldots,x_k\in B,\\ d(x_1,\ldots,x_k)<m_{i-1}}} \frac{h^{k-1}}{k!}|\rho_+(x_1,\ldots,x_k)|\\
&\le \biggl(\sum_{k=2}^\infty \frac{h^{2k-2}(2\ep)^{2(k-1)}m^{2k} }{k!}\biggr)^{1/2}\sqrt{s_0(B)}\\
&\le CK\beta\sqrt{v} m^2 \biggl(\sum_{k=2}^\infty \frac{(2\ep h m)^{2(k-1)}}{k!}\biggr)^{1/2}\\ 
&\le  CK \beta\sqrt{v}m^2\ep h m e^{2\ep^2h^2m^2}\le CK \beta \sqrt{v}m^2\frac{\sqrt{\log \log n}}{\log n}.
}
Combining the above steps, we get the claimed inequality. Retracing the steps, we get the same bound  for $F_-$.
\end{proof}

We are now ready to prove Lemma \ref{sqlmm}.
\begin{proof}[Proof of Lemma \ref{sqlmm}]
Let $h$ and $B$ be as in Lemma \ref{remainstep}. Consider the RFIM on $\Lambda$ with plus boundary condition. Modify the model as in the paragraph preceding Lemma \ref{taylorstep}, by adding $h$ to $\phi_x$ for all $x\in B$. Then, further modify the model by removing the links between $B$ and $\Lambda \setminus B$. Let $G_+(h)$ be the free energy of the resulting model after these two modifications. Then 
\[
G_+(h) = G_0(h) + R,
\]
where $G_0(h)$ is the free energy of the RFIM on $B$ with zero boundary condition and $\phi_x$ replaced by $\phi_x+h$ in the Hamiltonian, and $R$ is the free energy of the RFIM on $\Lambda\setminus B$ which has plus boundary condition on the part of $\partial (\Lambda \setminus B)$ that lies outside $B$,  and zero boundary condition on the part of $\partial(\Lambda \setminus B)$ that belongs to $B$. Note that $R$ does not depend on $h$. Thus,
\[
G_+(h)-G_+(0)= G_0(h)-G_0(0). 
\]
On the other hand, by the straightforward inequality
\[
\biggl|\log \sum_{\sigma} e^{-\beta H_1(\sigma)} - \log \sum_{\sigma}e^{-\beta H_2(\sigma)}\biggr|\le \beta \max_\sigma|H_1(\sigma)-H_2(\sigma)|
\]
that holds for any two Hamiltonians $H_1$ and $H_2$, and the fact that we are deleting at most $4m$ links, it follows that $|F_+(h)-G_+(h)|\le 4\beta m$ for any~$h$. Thus,
\eq{
&|(F_+(h)-F_+(0))- (G_0(h)-G_0(0))|\\
 &= |(F_+(h)-F_+(0))- (G_+(h)-G_+(0))|\le 8\beta  m. 
}
Lastly, observe that
\[
F_+'(0) = \beta \sqrt{v}\sum_{x\in B} \smallavg{\sigma_x}_+, 
\]
where $\smallavg{\sigma_x}_+$ is the quenched expectation of $\sigma_x$ in our original RFIM on $\Lambda$ with plus boundary condition. 
Combining the above steps and applying Lemma \ref{remainstep},  we get
\eq{
&\biggl|\ee\biggl(\sum_{x\in B} \smallavg{\sigma_x}_+\biggr) - \frac{\ee(G_0(h)-G_0(0))}{\beta \sqrt{v} h}\biggr| \\
&\le \frac{Cm^2}{\sqrt{v\log \log n}} + CKm^2 \frac{\sqrt{\log\log n}}{(\log n)^{1/4}}\le \frac{C(1+v^{-1/2})m^2}{\sqrt{\log \log n}}.
}
Let $\Lambda_1$ be the union of all $B\in \mb_0$. Let
\[
\theta := \frac{|\mb_0|\ee(G_0(h)-G_0(0))}{\beta \sqrt{v} h}.
\]
Then the above inequality implies that
\eq{
\biggl|\ee\biggl(\sum_{x\in \Lambda_1} \smallavg{\sigma_x}_+\biggr) - \theta\biggr|&\le \frac{C(1+v^{-1/2})n^2}{\sqrt{\log \log n}}.
}
By \eqref{l0size} and \eqref{mb1}, 
\[
|\Lambda\setminus \Lambda_1|\le \frac{2n^2}{K^2} + 2n^{3/2}\log n\le \frac{Cn^2}{(\log n)^{1/6}}. 
\]
Thus,
\eq{
\biggl|\ee\biggl(\sum_{x\in \Lambda} \smallavg{\sigma_x}_+\biggr) - \theta\biggr|&\le \frac{C(1+v^{-1/2})n^2}{\sqrt{\log \log n}} + \frac{Cn^2}{(\log n)^{1/6}}\\
&\le \frac{C(1+v^{-1/2})n^2}{\sqrt{\log \log n}}.
}
Proceeding exactly as above but with minus boundary condition, we get the same inequality for $\smallavg{\sigma_x}_-$, with the same $\theta$. Thus,
\[
\biggl|\ee\biggl(\sum_{x\in \Lambda} (\smallavg{\sigma_x}_+- \smallavg{\sigma_x}_-)\biggr)\biggr|\le \frac{C(1+v^{-1/2})n^2}{\sqrt{\log \log n}}.
\]
This completes the proof.
\end{proof}

Finally, we are ready to prove Theorem \ref{mainthm}. We will now revert back to the setting of Theorem \ref{mainthm}, where $\Lambda$ is an arbitrary finite subset of $\zz^2$ instead of a square. 
\begin{proof}[Proof of Theorem \ref{mainthm}]
It suffices to prove the theorem assuming that $\beta\in (0,\infty)$, because the case $\beta=0$ is trivial, and the inequality for $\beta=\infty$ can be deduced by taking a limit after we have proved the theorem for finite $\beta$, since the upper bound does not depend on $\beta$ and $\Lambda$ is a finite set (which implies that $\smallavg{\sigma_x}_+$ is a continuous function of $\beta$ as $\beta$ varies in $[0,\infty]$).

Let $\Lambda'$ be an $(n-1)\times (n-1)$ square containing $x$. Then $\Lambda'\subseteq \Lambda$. Let $\smallavg{\sigma_x}_{\Lambda,+}$ be the quenched expectation of $\sigma_x$ under plus boundary condition on $\Lambda$. Similarly, $\smallavg{\sigma_x}_{\Lambda',+}$ be the quenched expectation of $\sigma_x$ in the RFIM on $\Lambda'$ with plus boundary condition. The point $x$ can be made to take any position within the square $\Lambda'$ by choosing $\Lambda'$ suitably. Thus, Lemma \ref{sqlmm} implies that there exists some $\Lambda'$ as above, for which
\[
\ee(\smallavg{\sigma_x}_{\Lambda',+} - \smallavg{\sigma_x}_{\Lambda',-})\le \frac{C(1+v^{-1/2})}{\sqrt{\log \log n}},
\]
where $C$ is some universal constant. By \eqref{sigma2}, 
\eq{
\smallavg{\sigma_x}_{\Lambda',+} - \smallavg{\sigma_x}_{\Lambda',-} &\ge \smallavg{\sigma_x}_{\Lambda,+} - \smallavg{\sigma_x}_{\Lambda,-}\ge 0,
}
and  by \eqref{sigma1}, for any $\gamma, \gamma'\in \Gamma$,
\eq{
|\smallavg{\sigma_x}_{\Lambda,\gamma}- \smallavg{\sigma_x}_{\Lambda,\gamma'}|\le \smallavg{\sigma_x}_{\Lambda,+} - \smallavg{\sigma_x}_{\Lambda,-}.
}
This completes the proof.
\end{proof}

\section*{Acknowledgments}
I am grateful to Michael Aizenman and Ron Peled for motivating discussions, and to Hugo Duminil-Copin for checking the proof. I thank the referee for several useful suggestions.

\end{document}